\documentclass[dvips,aos,preprint]{imsart}

\usepackage{amsthm,amsmath}
\RequirePackage{natbib}
\RequirePackage[dvips]{hyperref}
\RequirePackage{hypernat}
\usepackage{amssymb,enumerate,epsfig}

\startlocaldefs
\numberwithin{equation}{section}
\theoremstyle{plain}

\newtheorem{assumption}{Assumption}[section]
\newtheorem{remark}{Remark}[section]
\newtheorem{claim}{Claim}[section]

\newtheorem{condition}{Condition}[section]
\newtheorem{criterion}{Criterion}[section]
\newtheorem{definition}{Definition}[section]
\newtheorem{example}{Example}[section]
\newtheorem{proposition}{Proposition}[section]
\newtheorem{problem}{Problem}[section]
\newtheorem{lemma}{Lemma}[section]

\endlocaldefs

\begin{document}

\begin{frontmatter}
\title{Nonlinear Principal Components and Long-run Implications
of Multivariate Diffusions\protect\thanksref{T1}}
\runtitle{Nonlinear Principal Components and the Long Run}
\thankstext{T1}{Forthcoming in the Annals of Statistics. The title of an earlier version of this paper is Principal Components and the Long Run.
We thank Anna Aslanyan, Henri
Berestycki, Jaroslav Borovi\v{c}ka, Ronald Coifman, E. Brian Davies,
Nan Li, Oliver Linton, Peter McCullagh, Nour Meddahi, Gigliola
Staffilani, Stephen Stigler, Gauhar Turmuhambetova and Noah Williams
for useful conversations.}

\begin{aug}
\author{\fnms{Xiaohong} \snm{Chen}\thanksref{t2}\ead[label=e1]{xiaohong.chen@yale.edu}},
\author{\fnms{Lars Peter} \snm{Hansen}\thanksref{t3}\ead[label=e2]{lhansen@uchicago.edu}}
\and
\author{\fnms{Jos\'{e}} \snm{Scheinkman}\thanksref{t4}
\ead[label=e3]{joses@princeton.edu}
}

\thankstext{t2}{National Science Foundation support under Award Number SES0631613.}
\thankstext{t3}{National Science Foundation support under Award Number SES0519372.}
\thankstext{t4}{National Science Foundation support under Award Numbers SES0350770 and SES0718407}
\runauthor{Chen, Hansen and Scheinkman}

\affiliation{Yale University, University of Chicago and Princeton
University}

\address{Xiaohong Chen\\ Cowles Foundation for Research in Economics\\
Department of Economics \\
Yale University \
  Box 208281 \\
  New Haven, CT 06520-8281\\
\printead{e1}}

\address{Lars Peter Hansen\\ Departments of Economics and Statistics\\
The University of Chicago\\
1126 East 59th Street\\
Chicago, IL 60637\\
\printead{e2}}

\address{Jos\'{e} Scheinkman \\ Department of Economics\\
Princeton University\\
26 Prospect Avenue \\
Princeton, NJ 08540-5296\\
\printead{e3}}

\end{aug}

\begin{abstract}
We investigate a method for extracting nonlinear principal
components (NPCs). These NPCs maximize variation subject to
smoothness and orthogonality constraints; but we allow for a general
class of constraints and multivariate probability densities,
including densities without compact support and even densities with
algebraic tails. We provide primitive sufficient conditions for the
existence of these NPCs. By exploiting the theory of
continuous-time, reversible Markov diffusion processes, we give a
different interpretation of these NPCs and the smoothness
constraints. When the diffusion matrix is used to enforce
smoothness, the NPCs maximize long-run variation relative to the
overall variation subject to orthogonality constraints. Moreover,
the NPCs behave as scalar autoregressions with heteroskedastic
innovations; this supports semiparametric identification and
estimation of a multivariate reversible diffusion process and tests
of the overidentifying restrictions implied by such a process from
low frequency data. We also explore implications for stationary,
possibly non-reversible diffusion processes. Finally, we suggest a
sieve method to
estimate the NPCs from discretely-sampled data.
\end{abstract}



\end{frontmatter}



\section{Introduction}

Principal components are functions of the data that capture maximal
variation in some sense.  Often they are restricted to be linear
functions of the underlying data as in original analyses of
\citet{pearson} and \citet{hotelling}.  In this paper we study the
extraction of nonlinear principal components (NPCs) using
information encoded in the probability density of the data.
Formally, the NPCs maximize variation subject to orthogonality and
smoothness constraints where smoothness constraints are enforced by
a quadratic form $f$ expressed in terms of the gradients of
functions. Specifically, the quadratic form is
\[
f(\phi,\psi) = {\frac 1 2}\int_{x\in \Omega} \nabla
\phi(x)'\Sigma(x) \nabla \psi(x)q(x) dx
\]
where $\nabla$ denotes the (weak) gradient operator, $\Omega$ is the
state space, $\Sigma$ is a state-dependent positive-definite matrix,
and $q$ is the invariant density of a strictly stationary ergodic
data $\{x_i\}_{i=1}^{T}$.

Alternatively, NPCs are solutions to approximation problems. Suppose
we wish to form the best finite-dimensional least squares
approximation to an infinite-dimensional space of smooth functions,
where we use the form $f$ to limit the class of functions to be
approximated.  In a sense that we make formal, a finite number of
NPCs solves this problem. More stringent smoothness restrictions
enforced by penalization limit the family of functions to be
approximated while improving the overall quality of approximation.
Thus our analysis of NPCs is in part an investigation of this
approximation.

Previously \citet{boxtiao} proposed a canonical analysis of
multivariate linear time series. This analysis produces linear
principal components of the multivariate process that can be ordered
from least to most predictable.  Much later in a seemingly unrelated
paper, \citet{salinelli} defined NPCs for multivariate absolutely
continuous random variables and characterized these NPCs as
eigenfunctions of a self-adjoint, differential operator. As we will
show these two methods are related. We share \citet{salinelli}
interest in NPCs, but our departure from his work is substantial.
For Salinelli, the matrix $\Sigma$ is the identity matrix,
the state space $\Omega$ is compact and the
density $q$ is bounded above and below for the bulk of his analysis.
Our interest in probability densities $q$ that do not have compact
support, including densities with algebraic tails, leads us
naturally to consider a more general class of smoothness penalties.
By allowing for a more flexible specification for $\Sigma$ and $q$,
we entertain a larger class of smoothness constraints {\it vis a
vis} \citet{salinelli} with explicit links to the data generation.
Establishing the existence of NPCs in our setup is no longer
routine.

\citet{salinelli} assumed that the data generation process is
independent and identically distributed (IID).  While our analysis
is applicable to such an environment, we also explore the case in
which data $\{x_i\}_{i=1}^{T}$ is sampled in low frequency from a
stationary Markov diffusion process. By considering such processes,
we make a specific choice of the matrix $\Sigma$ used to enforce
smoothness. It is the local covariance or diffusion matrix. With
this choice, the NPCs extracted with smoothness penalties are
ordered by the ratio of their long-run variation to the overall
variation as in \citet{boxtiao}. NPCs that capture variation subject
to smoothness constraints also display low frequency variation due
to their high persistence. In effect we provide an extension of the
method of \citet{boxtiao} to nonlinear, multivariate Markov
diffusions, and establish an explicit link to the method of
\citet{salinelli}.

In this paper we do the following:

\begin{enumerate}

\item{Formulate the NPCs extraction to include state
dependence in the smoothness
constraint and state spaces that have infinite Lebesgue measure.}

\item{Give sufficient conditions for the existence of these
NPCs.}

\item{Provide a reversible Markov diffusion process for the data generation that supports the
NPCs extraction method and generates testable implications.}

\item{Explore implications for a more general class of Markov
diffusion processes.}


\end{enumerate}

The rest of the paper is organized as follows. In section
\ref{sec:prin}, we first define NPCs as functions that maximize
variation subject to orthogonality conditions and smoothness bounds
given by the quadratic form $f$. Section \ref{sec:exist} presents
existence results. In section \ref{sec:forms} we suppose the data
are sampled from a multivariate nonlinear diffusion and establish
the connection between our NPCs and the canonical analysis of
\citet{boxtiao}. The results in section \ref{sec:reverse} relate the
NPCs to eigenfunctions of conditional expectations operators
associated with a stationary Markov process $\{x_t\}$ defined using
the diffusion matrix $\Sigma$ and the stationary density $q.$ Given
an eigenfunction $\psi,$ the process $\{\psi(x_t)\}$ behaves as a
scalar autoregression. Thus the eigenfunctions we obtain satisfy
testable implications when the data is generated by a
 Markov process. The Markov process constructed in Section \ref{sec:reverse}
 is time reversible.  In Section \ref{sec:irrev} we characterize other Markov processes
 associated with the same $q$ and $\Sigma$. Section
\ref{sec:consist} provides a sieve method to
estimate these NPCs using discrete-time low frequency observations
$\{x_i\}_{i=1}^{T}$. Section \ref{sec:conclude} gives some
concluding remarks and discusses applications of our results. The
appendices contain computations associated with an example and some
proofs that are not stated in the main text.

\section{Nonlinear principal components} \label{sec:prin}

To define a functional notion of principal components we require two
quadratic forms. We start with an open connected $\Omega \subseteq
\mathbb{R}^n$.  Let $q$ be a probability density on $\Omega$ with
respect to Lebesgue measure.  The implied probability distribution
$Q$ is the population counterpart to the empirical distribution of
the data.  The data could be IID as in \citet{salinelli}, but we are
primarily interested in the case in which the data
$\{x_i\}_{i=1}^{T}$ are sampled in low frequency from a
continuous-time, stationary Markov diffusion $\{ x_t :t\ge 0 \}$. In
this case $q$ is the stationary density of $x_t$.

Let $L^2$ denote the space of Borel measurable square integrable
functions with respect to the population probability distribution
$Q$. The $L^2$ inner product (denoted $<\cdot, \cdot>$) is one of
the two forms of interest. We use the corresponding norm to define
an approximation criterion.

The second form is used to measure smoothness.  Consider a
(quadratic) form $f_o$ defined on $C_K^2$, the space of twice
continuously differentiable functions with compact support in
$\Omega$,  that can be parameterized in terms of the density $q$
and a positive definite matrix $\Sigma$ that can depend on the
state:
\begin{equation}
f_o(\phi, \psi) = {\frac 1 2} \int_{\Omega} \sum_{i,j} \sigma_{ij}
{\frac {\partial \phi} {\partial y_j}}{\frac {\partial \psi}
{\partial y_i}}q , \label{edefpar}
\end{equation}
where
\[
\Sigma = [ \sigma_{ij} ].
\]

\begin{assumption}
\label{adensity} $q$ is a positive, continuously differentiable
probability density on $\Omega$.
\end{assumption}

\begin{assumption} \label{adiffusion}
$\Sigma $ is a continuously differentiable, positive definite matrix
function on $\Omega$.
\end{assumption}

\noindent Assumptions \ref{adensity} and \ref{adiffusion} restrict
the density $q$ and the matrix $\Sigma$ to be continuously
differentiable. These assumptions are made for convenience. As
argued by \citet{davies89} (see Theorem 1.2.5) these restrictions can
be replaced by a less stringent requirement that entries of the
matrix $q\Sigma$ are locally (in $L^2$(Lebesgue)), weakly
differentiable.

While the $f_o$ is constructed in terms of the product $q \Sigma$,
the density $q$ will play a distinct role when we consider extending
the domain of the form to a larger set of functions.

To study the case in which $\Omega$ is not compact, we will
consider a particular closed extension of the form $f_o$. We
extend the form $f_o$ to a larger domain ${\bar H}\subset {L}^2 $
using the  notion of a weak derivative.
\begin{eqnarray*}
{\bar H} & \doteq &\{\phi \in L^2: \mathrm{ there \ \ exists}\; g\;
\mathrm{measurable, \ \ with}\; \int g^{\prime }\Sigma g q <
\infty,\\ & & \mathrm{and}\; \int \phi \nabla \psi = - \int g \psi,
\mathrm {for \ \ all }\; \psi \in  C_K^{1} \}.
\end{eqnarray*}
The random vector $g$ is unique (for each $\phi$) and is referred to
as the {\it weak derivative} of $\phi$. From now on, for each $\phi$
in ${\bar H}$ we write $\nabla \phi = g$.

Notice that ${\bar H}$ is constructed exactly as a weighted Sobolev
space except that instead of requiring that $g \in L^2$, we require
that $\Lambda g \in L^2$ where $\Lambda$ is the square root of
$\Sigma$. Also we use $C_K^{1}$ test functions. One can show, using
mollifiers, that allowing for this larger set of test functions is
equivalent to using the more usual set of test functions,
$C_K^{\infty}$  (see \citet{brezis} Remark 1, page 150.) For any pair
of functions $\psi$ and $\phi$ in ${\bar H}$ we define:
\[
f(\phi,\psi) = {\frac 1 2} \int_{\Omega} ({\nabla \phi})'\Sigma
(\nabla \psi)q,
\]
which is an extension of $f_o$. In ${\bar H}$ we use the inner
product $<\phi, \psi>_{\bar f}= < \phi, \psi> + f(\phi,\psi)$. With
this inner product, ${\bar H}$ is complete and hence a Hilbert space
(see Proposition \ref{hilbertspace} in the Appendix). Thus ${\bar
H}$ is taken to be the domain ${\mathcal D}(f)$ of the form $f$.
 Notice, in particular, that the unit function is in ${\mathcal D}(f) = \bar H$.

\subsection{ Initial construction}
NPCs maximize variation subject to smoothness constraints. In our
generalization these NPCs are defined as follows.

\begin{definition}\label{d:sprina} The function $\psi_{j}$ is the $j^{th}$ nonlinear
principal component (NPC) for $j\geq 1$ if $\psi _{j}$ solves:
\[ \max_{\phi \in \bar H } <\phi, \phi>
\]
subject to
\begin{eqnarray*}
f(\phi, \phi) & = & 1,  \\
<\psi_s,\phi> & = & 0, s=0,...,j-1,
\end{eqnarray*}
where $\psi_{0}$ is initialized to be the constant function one.
\end{definition}

There are two differences between our proposed extraction and that
of \citet{salinelli}.
First, \citet{salinelli} assumes that $\Sigma$ is the identity matrix.
To accommodate a richer class of densities,
we allow $\Sigma$ to be state dependent. Second, \citet{salinelli}
assumes that the data density $q$ has finite Lebesgue measure and is
bounded away from zero. We allow the Lebesgue measure of the state
space to be infinite, and accordingly our density $q$ is no longer
assumed to be bounded from below.


NPCs are eigenfunctions of the quadratic forms $f$.
\begin{definition}
An eigenfunction $\psi$ of the quadratic form $f$ satisfies:
\begin{equation}
\label{eigenform} f(\phi,\psi) = \delta <\phi,\psi>
\end{equation}
for all $\phi \in {\mathcal D}(f)$. The scalar $\delta$ is the
corresponding eigenvalue.
\end{definition}

\noindent Since $f$ is positive semidefinite, $\delta$ must be
nonnegative.  The NPCs extracted in the manner given in
(\ref{d:sprina}) have eigenvalues $\delta_j$ that increase with $j$.
If we renormalize the eigenfunctions to have a unit second moment,
the NPCs will be ordered by their smoothness as measured by
$\delta_j = f({\psi}_j,{\psi}_j )$.  Moreover, $f(\psi_j,\psi_k) =
0$ for $j \ne k$.

Suppose that the  NPCs $\{ \psi_j : j = 0, 1, ... \}$ exist with
corresponding eigenvalues $\{ \delta_j : j = 0,1,... \}.$ Consider
any $\phi$ in $L^2$. Then
\[
\phi = \sum_{j=0}^\infty {\frac {<\psi_j,\phi>} {<\psi_j, \psi_j>}}
\psi_j,
\]
\[
<\phi, \phi> =  \sum_{j=0}^\infty {\frac {<\psi_j, \phi>^2}
{<\psi_j,\psi_j>}},
\]
and for any $\phi,\psi \in {\mathcal D}(f)$,
\begin{equation} \label{eigen:fdecom}
f(\phi, \psi) = \sum_{j=0}^\infty \delta_j {\frac {<\phi,
\psi_j><\psi,\psi_j>}{<\psi_j,\psi_j>}} .
\end{equation}

\subsection{Benchmark optimization problem} \label{subsec:bench}

Let $H$ be a closed linear subspace of $L^2$, and consider the
optimization problem:

\begin{problem} \label{prob:jose1}
\[
\max_{\phi \in H} < \phi, \phi>
\]
subject to
\[
\theta <\phi, \phi> +  f(\phi, \phi) \le 1
\]
for some $\theta > 0$.
\end{problem}

A necessary condition for $\psi$ to be a NPC is that it satisfies an
eigenvalue problem:

\begin{claim} \label{claim:eigen}
A solution $\psi$ to Problem \ref{prob:jose1} will also solve the
eigenvalue problem:
\[
<\phi, \psi> = \lambda[\theta <\phi, \psi> + f(\phi, \psi )]
\]
for some positive $\lambda$ and all $\phi \in H$.
\end{claim}

To establish the existence of a solution to Problem
\ref{prob:jose1}, it suffices to suppose the following:
\begin{condition}
\label{condhigh} (existence) $\{\phi \in \mathcal{D}(f) :  f(\phi,
\phi) + \theta <\phi,\phi> \le 1\}$ is precompact (has compact
closure) in $L^2$.
\end{condition}
\noindent  The precompactness restriction guarantees that we may
extract an $L^2$ convergent sequence in the constraint set, with
objectives that approximate the supremum.  The limit point of
convergent sequence used to approximate the supremum, however, will
necessarily be in the constraint set because the constraint set is
convex and the form is closed.

\subsection{Approximation}

Why do we care about NPCs?  One way to address this is to explore
the construction of the best, finite-dimensional, least squares
approximations. Specifically, suppose we wish to construct the best
finite dimensional set of approximating functions for the space of
functions that are square integrable with respect to a probability
measure $Q$ with density $q$. We now motivate NPCs as the recursive
solution to such a problem. The $N$-dimensional problem is solved by
solving $N$ one-dimensional problems using a sequence of $H$'s that
remove one dimension in each step.  The outcome at each step is a
NPC used as an additional approximating function.

Initially solve Problem \ref{prob:jose1} for $H = L^2$, select a
solution $\psi_0$ and denote the maximized objective as $\lambda_0$.
Inductively, given $\psi_0, \psi_1, ...,\psi_{j-1}$, form $H_{j-1}$
as the $j$ dimensional space generated by these $j$ solutions
constructed recursively. Let $H_{j-1}^\perp$ denote the space of all
elements of $L^2$ that are orthogonal to these $j$ solutions and
hence orthogonal to $H_{j-1}$.  Solve Problem \ref{prob:jose1} for
$H = H_{j-1}^\perp$, select a solution $\psi_j$, and form
$\lambda_j$ as the maximized value. The sequence $\{ \lambda_j :
j=0,1,...\}$ is decreasing because we are omitting components of the
constraint set for the maximization problem as $j$ increases.

In what sense is such a recursive procedure optimal?    In answering
this question, let $Proj(\phi|{\hat H})$ denote the least squares
projection of $\phi$ onto the closed (in $L^2$) linear space ${\hat
H}$.  The second moment of the approximation error is:
\[
<\phi - Proj(\phi|{\hat H}), \phi - Proj(\phi|{\hat H})> =
<\phi,\phi> - [Proj(\phi| {\hat H})]^2.
\]

\begin{claim} \label{claim:bound} Let ${\hat H}$ denote any $N$-dimensional subspace
of $L^2$.  Then
\[
 \max_{\{\phi:\theta <\phi, \phi> + f(\phi, \phi) \le 1\}} \{ <\phi, \phi> -
 [Proj(\phi| {\hat H})]^2 \}\ge \lambda_{N}.
\]
\end{claim}

Our next result shows that the bound deduced in Claim
\ref{claim:bound} is attained by $H_{N-1}.$

\begin{claim} \label{claim:attain}
\[
\max_{\{\phi: \theta <\phi, \phi> + f(\phi, \phi) \le 1\}} \{<\phi,
\phi> - [Proj(\phi| H_{N-1})]^2\} = \lambda_N .
\]
\end{claim}
%

Taken together, these two claims justify $H_{N-1}$ as a good
$N$-dimensional space of approximating functions.

\begin{remark}
There exist $N$-dimensional spaces other than $H_{N-1}$ that attain
the bound given in Claim \ref{claim:bound}. One reason is that there
may be multiple solutions to Problem \ref{prob:jose1}. Even when the
solution to Problem \ref{prob:jose1} is unique, at each stage of the
construction $\psi_{N-1}$ may be replaced by the sum of $\psi_{N-1}$
plus some $\psi's$ that is orthogonal to all of the solutions to
Problem \ref{prob:jose1} with $H = H_{N-1}^\perp$. Such a choice
cannot necessarily be used in a recursive construction of optimal
approximating spaces with dimension greater than $N$.
\end{remark}

\subsection{Nonlinear principal components revisited}

In Problem \ref{prob:jose1}, the constraint set gets larger as
$\theta$ declines to zero.  Reducing the smoothness penalty with a
smaller $\theta$ enlarges the collection of functions that satisfy
the constraint.  Thus the maximized objective increases as $\theta$
is reduced. While this is true, it turns out the maximizing choice
of $\phi$ does not depend on $\theta$ up to scale. This follows
because the ranking over $\phi$'s implied by the ratio:
\[
{\frac {<\phi, \phi>} {\theta <\phi,\phi> + f(\phi,\phi)}}
\]
does not depend on the value of $\theta$.  The same ranking is
also implied by the ratio:
\[
{\frac {<\phi, \phi>} {f(\phi,\phi)}}
\]
provided that $H$ is orthogonal to all constant functions.  Thus a
scaled solution $\psi$ to Problem \ref{prob:jose1} also
solves:

\begin{problem}
\[
\max_{\phi \in H} <\phi,\phi>
\]
subject to:
\[
f(\phi,\phi) = 1.
\]
\end{problem}
Restricting $H$ to be orthogonal to constant functions is equivalent
to limiting attention to functions $\phi$  that have mean zero under
the population data distribution $Q$. Recall that our construction
of NPCs is based on the recursive application of this problem.

From Claim \ref{claim:eigen} we know that $\psi$ satisfies:
\[
<\phi,\psi> = \lambda[\theta <\phi,\psi> + f(\phi, \psi )]
\]
for all $\phi \in H$.  Rearranging terms,
\[
f(\phi, \psi ) =  \delta <\phi,\psi >
\]
where
\[
\delta = { \frac  {1 - \theta \lambda}{\lambda }}.
\]
This is the eigenvalue associated with the NPC extraction.  Solving
for $\lambda$,
\[
\lambda = {\frac 1 {\theta + \delta }}.
\]
Since eigenvalues $\delta$ of the form increase without bound, the
corresponding sequence of $\lambda$'s converge to zero guaranteeing
that approximation becomes arbitrarily accurate as the number of
NPCs increases.

\section{Existence} \label{sec:exist}

In this section we consider more primitive sufficient conditions
that imply Condition \ref{condhigh}, which as we noted in  section
\ref{sec:prin}, guarantees the existence of NPCs.
 We allow for noncompact
state spaces and provide alternative restrictions on the tail
behavior of the the density $q$ and the penalization matrix $\Sigma$
that guarantee that the compactness criterion (Condition
\ref{condhigh}) is satisfied. Roughly speaking when the tails of the
density $q$ are exponentially thin, the compactness criterion can be
established without requiring that the matrix $\Sigma$ becomes large
(in the sense of positive definite matrices) in the tails. On the
other hand, when the tails of $q$ are algebraic and hence thicker,
divergence of $\Sigma$ in the tails can play an important role in
establishing Condition \ref{condhigh}.

We start by reviewing some known existence conditions, which we
extend using two devices. First, we transform the function space and
hence the (quadratic) form so that distribution induced by $q$ is
replaced by the Lebesgue measure. This transformation allows us to
apply known results for forms built using the Lebesgue measure. Second,
we study forms that are simpler but dominated by $f$. When the
dominated forms satisfy Condition \ref{condhigh} the same can be
said of $f$.

\subsection{Compact Domain}

\citet{salinelli} established the existence of eigenfunctions by
 applying Rellich's compact embedding theorem when the domain $\Omega $ is compact
 with a continuous boundary. This
approach requires a density $q$ that is bounded and bounded away
from zero and a derivative penalty  matrix $\Sigma$ that is
uniformly nonsingular.

\subsection{Real Line}

Perhaps surprisingly, the NPC extraction is nontrivial even for
densities on the real line.  This is because our NPCs can be
nonlinear functions of the underlying Markov state. We initially consider the case in which the
state space is the real line.

\begin{proposition} \label{scalar}
Suppose $ \Sigma = \varsigma^2 $ and
\begin{equation} \label{compactcrita}
 \int_0^\infty {\frac 1 {\varsigma^2(x) q(x)}} = + \infty ,
\int_0^\infty {\frac 1 { \varsigma^2(-x) q(-x)}} = + \infty ,
\end{equation}
\begin{equation} \label{compactcritb} \lim_{|x|
\rightarrow \infty}  - {\frac x {|x|}} \left[  \varsigma(x) {\frac
{ q^{\prime}(x)} {q(x)}} + \varsigma^{\prime}(x) \right] = +
\infty .
\end{equation}
Then Condition \ref{condhigh} is satisfied.
\end{proposition}

When $\varsigma$ is constant, the compactness condition
(\ref{compactcritb}) reduces to:
\[
\lim_{|x| \rightarrow  \infty}  - {\frac x {|x|}} \left[ {\frac {
q^{\prime}(x)} {q(x)}}\right] = + \infty,
\]
which rules out densities with algebraic tails (tails that decay
slower than $|x|$ raised to a negative power.)  By allowing for
$\varsigma$ to increase, we can accommodate densities with
algebraic tails.  We now extend this analysis to higher
dimensions.

\subsection{${\mathbb R}^n$} \label{subsec:whole}

In the subsections that follow, we will provide multivariate
extensions for both sources of compactness: growth in the
logarithmic derivative of the density $q$ and growth in the
derivative penalty $\Sigma$.  For simplicity, we will concentrate
in the case where the state space is all of $\mathbb{R}^n.$

\subsubsection{Cores} \label{subsect:core}

The compactness Condition \ref{condhigh} involves the domain of
the form $f$ which is often rather complicated to describe. For
this reason, we will focus on cases where this domain can be well
approximated by smooth functions. The adequate notion of
approximation is that of a core:

\begin{definition}  A
family of functions $\mathcal{C}o \subset \mathcal{D}(f)$ is a
{\bf core} of $f$ if for any $\phi_0$ in the domain
$\mathcal{D}(f)$, there exists a sequence $\{\phi_j\}$ in
$\mathcal{C}o$ such that
\[
\lim_{j \rightarrow \infty} <\phi_j-\phi_0, \phi_j-\phi_0> +
f(\phi_j - \phi_0, \phi_j - \phi_0) = 0.
\]
\end{definition}

\begin{condition} \label{cond:core} $C_K^2$ is a core of $f$.
\end{condition}

Let ${\hat f}$ denote the {\em minimal extension}, the smallest
closed extension of the form $f_o$ defined in equation
(\ref{edefpar}). Condition \ref{cond:core} is equivalent to $f=
{\hat f}.$

Although their purpose was different, \citet{fot} provide a
convenient sufficient condition that implies Condition
\ref{cond:core} in environments that interest us.  Define:
\[
\kappa(r) = \int_{|x|=1} x'\Sigma({\sf r}x) x q({\sf r}x) dS(x)
\]
where $dS$ is the measure (surface element) used for integration
on the sphere $|x|=1$. For functions $\psi$ and $\phi$ in
$C_{K}^{2}$ that are radially symmetric, {\it i.e.} $\phi(x) =
\xi(|x|)$ and $\psi(x) = \zeta(|x|)$, we may depict the form $f_o$
as an integral over radii:
\[
f_o(\psi,\phi) = \int_0^\infty {\frac {d \xi({\sf r})} {d{\sf r}}}
{\frac {d \zeta({\sf r})} {d{\sf r}}} \kappa({\sf r}) {\sf
r}^{n-1} dr.
\]

\begin{proposition} \label{prop:recur} Condition \ref{cond:core} is implied by:
\begin{equation}
\label{nattract} \int_1^\infty \kappa({\sf r})^{-1} {\sf
r}^{1-n}d{\sf r} = \infty .
\end{equation}
\end{proposition}

\noindent Restriction (\ref{nattract}) implies the scalar
restriction (\ref{compactcrita}) of Proposition \ref{scalar}. This
follows since for any non-negative reals ${\sf r}_1$ and ${\sf
r}_2$,
\[
\min \left\{ {\frac 1 {{\sf r}_1}}, {\frac 1 {{\sf r}_2}}\right\}
\ge {\frac 1 {{\sf r}_1+{\sf r}_2}}.
\]

Notice that (\ref{nattract}) is a joint restriction on $\Sigma$
and $q$.  We may relate this condition to the moments of $q$ and
the growth of $\Sigma$ using the inequality:
\begin{equation*}
\infty  = \left(\int_1^\infty {\frac 1 {\sf r}} d{\sf r} \right)^2
\le  \int_1^\infty \kappa({\sf r})^{-1} {\sf r}^{1-n}dr
\int_1^\infty \kappa({\sf r}) {\sf r}^{n-3}dr.
\end{equation*}
Thus a sufficient condition for (\ref{nattract}) is that
\begin{equation}
\label{nattract2} \int_1^\infty {\frac {\kappa({\sf r})} {{\sf
r}^2}} {\sf r}^{n-1} d {\sf r} < \infty.
\end{equation}

This latter inequality displays a tradeoff between growth in the
penalization matrix and moments of the distribution. Define
\[
\varsigma^2({\sf r}) = \sup_{|x|=1} x' \Sigma({\sf r}x) x,
\]
and
\[
\varrho({\sf r}) = \int_{|x| = 1} q({\sf r}x) dS(x).
\]
Notice that
\[
\kappa({\sf r}) \le \varsigma^2({\sf r}) \varrho({\sf r}).
\]
Suppose for instance, $\varsigma^2({\sf r})$ is dominated by a
quadratic function (in ${\sf r}$). Then (\ref{nattract2}) and
hence (\ref{nattract}) are satisfied because the density $q$ is
integrable:
\[
\int_0^\infty \varrho({\sf r}) {\sf r}^{n-1} d{\sf r} = 1.
\]

We may extend the previous argument by supposing instead that
\[
\varsigma^2({\sf r}) \le c |{\sf r}|^{2+2\delta}
\]
for some positive $\delta$.  Then
\[
{\frac {\kappa({\sf r})} {{\sf r}^2}} \le c {\sf r}^{2\delta}
\int_{|x| = 1} q({\sf r}x) dS(x).
\]
Thus (\ref{nattract2}) is satisfied provided that
\[
\int |x|^{2\delta} q(x) dx < \infty.
\]
Hence we can allow for faster growth in $\varsigma^2$ if $q$ has
high enough moments.

So far we have produced a sufficient condition for approximation
using functions in $C^2_K$ (Condition \ref{cond:core}).  We provide sufficient
conditions for the original compactness condition (Condition \ref{condhigh}) by transforming
the probability measure.

\subsubsection{Transforming the Measure}

In this subsection we map the original probability  space $L^2$
into a Lebesgue counterpart $L^2(leb).$   The transformation is
standard (see \citet{davies89}), but it is often applied in the
reverse direction. By using this transformation we may appeal to
some existing mathematical results on compactness to establish
Criterion \ref{condhigh},
\[
{\mathcal U}_\theta = \{\phi \in {\mathcal D}(f) : f(\phi,\phi) +
\theta <\phi,\phi>  \le 1 \}
\]
is precompact in $L^2$ for some $\theta > 0$.

Given $q$ write:
\[
q^{1/2} = \exp(-h).
\]

\begin{assumption} \label{adensitys}
The function $h$ is twice continuously differentiable.
\end{assumption}
\noindent This assumption imposes some extra smoothness on the
density, that was not required in our previous analysis.

Map the space $L^2$ into $L^2(leb)$ by the (invertible) unitary
transformation:
\[
\psi = U \phi \equiv \exp(-h) \phi.
\]
Since $U$ is unitary, it suffices to show that $U({\mathcal
U}_\theta)$ is pre-compact. We will actually construct a set that
contains $U({\mathcal U}_\theta)$ and is pre-compact in
$L^2(leb)$.

First notice that  $U$ and $U^{-1}$ leave $C_K^2$ invariant, and
for any $\psi \in C^2_K$ the corresponding $\phi=U^{-1} \psi$
satisfies:

\[
\nabla \phi =  \exp(h)(\psi \nabla h  + \nabla \psi ).
\]
Thus
\begin{eqnarray*}
{f}(U^{-1}\psi,U^{-1}\psi^*) &  = & {\frac 1 2} \int (\nabla
\psi)'
 \Sigma (\nabla \psi^*)
+  {\frac 1 2} \int (\nabla h)' \Sigma [\nabla (\psi \psi^*)] \cr
&& + {\frac 1 2} \int (\nabla h)' \Sigma (\nabla h) \psi \psi^*
\end{eqnarray*}
Applying integration-by-parts to $\psi \in C_K^2$, it follows that
\[
 \int (\nabla h)' \Sigma [\nabla (\psi \psi^*)] = -
 \int \sum_{i,j} \sigma_{i,j} {\frac {\partial^2 h}
{\partial y_i \partial y_j}} \psi \psi^* -  \int \sum_{i,j} {\frac
{\partial \sigma_{i,j}} {\partial y_i}} {\frac {\partial h}
{\partial y_j}} \psi \psi^*.
\]
Therefore,
\begin{equation}
\label{formn} {f}(U^{-1}\psi,U^{-1}\psi^*) = {\frac 1 2} \int
(\nabla \psi)' \Sigma (\nabla \psi^*) + {\frac 1 2} \int V \psi
\psi^*
\end{equation}
where the {\it potential} function $V$ is given by:
\begin{equation}
\label{potentialf} V = -  \sum_{i,j} \sigma_{i,j} {\frac
{\partial^2 h} {\partial y_i \partial y_j}} - \sum_{i,j} {\frac
{\partial \sigma_{i,j}} {\partial y_i}} {\frac {\partial h}
{\partial y_j}} + (\nabla h)' \Sigma (\nabla h).
\end{equation}

\begin{proposition} \label{prop:weaktrans}
Suppose that $C_K^2$ is a core for $f$, $\psi = U \phi$ for some
$\phi \in {\bar H}$ and $V$ is bounded from below.  Then $\psi$ is
weakly differentiable,
\[
\nabla \psi =  \exp(-h)(-\phi \nabla h  + \nabla \phi )
\]
and
\begin{equation} \label{weakbound} {\frac 1 2} \int (\nabla \phi)' \Sigma  \nabla \phi
q =  {\frac 1 2}\int (\nabla \psi)'\Sigma (\nabla \psi)+ {\frac 1
2}\int V \psi^2 .\end{equation}
\end{proposition}

A consequence of this proposition is that
\[
{\mathcal V}_\theta = \{ \psi \in L^2(leb) : \int \left( \theta  +
{\frac 1 2} V \right) \psi^2 + {\frac 1 2} \int (\nabla
\psi)'\Sigma (\nabla \psi) \le 1 \} \supset U({\mathcal
U}_\theta),
\]
and it thus suffices to show that ${\mathcal V}_\theta$  is
precompact in $L^2(leb)$ for some $\theta>0$.

We consider two methods for establishing that this last property
is satisfied. We first focus on the behavior of the potential $V$
used in the quadratic form: $ \int (\theta + {\frac 1 2} V)
\psi^2$, and then we study extensions that exploit growth in the
derivative penalty matrix $\Sigma$ used in the quadratic form:
$\int (\nabla \psi)'\Sigma (\nabla \psi)$.

\subsubsection{Divergent Potential}

In this section, we use the tail behavior of the potential $V$. To
simplify the treatment of the term $\int (\nabla \psi)'\Sigma
(\nabla \psi)$ in the definition of ${\mathcal V}_\theta$   we
impose:

\begin{assumption} \label{adiffusionb}
The derivative penalty matrix $\Sigma \ge {\underline {\sf c}} I$
for some ${\underline {\sf c}} > 0$.
\end{assumption}
\noindent This assumption rules out cases in which the derivative
penalty matrix diminishes to zero for arbitrarily large states.

We also suppose that the potential function diverges at the
boundary:

\begin{criterion} \label{ccompact}
$\lim_{|x|\rightarrow \infty}  V(x) = + \infty$.
\end{criterion}

\begin{proposition} \label{cllevel}
Under Assumptions \ref{adensitys} and \ref{adiffusionb}, if
Criterion \ref{ccompact} is satisfied, then Condition
\ref{condhigh} is satisfied.
\end{proposition}

Direct verification of Criterion \ref{ccompact} may be difficult
because formula (\ref{potentialf}) is a bit complicated. However,
we may replace the $\Sigma$ by a lower bound. Given Assumption
\ref{adiffusionb}  we can always construct a twice continuously
differentiable function $\varsigma(x)$ with
\begin{equation}\Sigma(x) \ge \varsigma(x)^2 I \ge {\underline {\sf c}}
I,~~ for~ some~ {\underline {\sf c}} >
0.\label{cpotentiald}
\end{equation}
We now show how growth conditions
on $\varsigma(x)$ can help in delivering compactness.

Let:
\[L
{\check f_o}(\phi, \phi^*) = {\frac 1 2} \int  \nabla \phi(x)
\cdot \nabla \phi^*(x) \varsigma(x)^2 q(x)
\]
on the space $\mathcal{C}_K^2$. Then
\[
{\check f_o}(\phi, \phi) \le f_o(\phi, \phi).
\]
Let ${\check f}$ be the minimal extension of ${\check f_o}.$ If
$f$ is the minimal extension of $f_o,$ when equation
(\ref{cpotentiald}) holds, the domain of ${\check f}$ contains the
domain of $f.$ Applying Proposition \ref{prop:weaktrans} to
${\check f},$ it suffices to use
\[
{\check V}(x) =   \varsigma(x)^2\left( -{\rm trace}\left[{\frac
{\partial^2 h(x)} {\partial x_i \partial x_j}}\right] - {\frac {2
 \nabla \varsigma(x) \cdot \nabla h(x)} {\varsigma(x)}}  + |\nabla
h(x)|^2 \right).
\]
 in place of $V$ in demonstrating
compactness.

\begin{criterion} Equation (\ref{cpotentiald}) is satisfied and
\label{ccompact1}
\[
\lim_{|x| \rightarrow \infty} {\check V}(x) = +\infty.
\]
\end{criterion}

To derive some sufficient conditions for this criterion we
parameterize $\varsigma$ as:
\[
\varsigma(x) = \exp[v(x)].
\]
Then an alternative formula for ${\check V}$ is:
\begin{eqnarray*} {\check V}(x) & = &
-\varsigma(x)^2 {\rm trace}\left[{\frac {\partial^2 h(x)}
{\partial x_i \partial x_j}}\right] + \varsigma(x)^2| \nabla h(x)
- \nabla v(x)|^2 - \varsigma(x)^2 \nabla v(x) \cdot \nabla v(x).
\end{eqnarray*}
An alternative to Criterion \ref{ccompact1} is:

\begin{criterion}  \label{compact1} Equation (\ref{cpotentiald}) is satisfied with $\varsigma(x) = \exp[v(x)]$ and,
\begin{description}
\item[a)]
 \[
 \lim_{|x| \rightarrow \infty} {\frac {|\nabla
v(x)|} {|\nabla h(x)|}}= 0;
\]
\item[b)]
\[\
 \lim_{|x| \rightarrow \infty} \varsigma(x)^2\left( -{\rm
trace}\left[{\frac {\partial^2 h(x)} {\partial x_i \partial
x_j}}\right] +  \nabla h(x) \cdot \nabla h(x) \right) = + \infty.
\]
\end{description}
\end{criterion}

\begin{proposition} Suppose Assumptions \ref{adensitys} is satisfied.  Then
Criterion \ref{compact1} implies Condition \ref{condhigh}.
\end{proposition}

\noindent Restriction b) of Criterion \ref{compact1} limits the
second derivative contribution from offsetting that of the squared
gradient of $h$. This criterion is convenient to check when $h$
displays polynomial growth, or equivalently when $q$ has
exponentially thin tails. Even if $|\nabla h|$ becomes arbitrarily
small for large $|x|$, the compactness criterion can still be
satisfied by having the penalization $\varsigma$ increase to more
than offset this decline.

Next we consider a way to exploit further growth in $\nabla
\varsigma$.  This approach gives us a way to enhance the potential
function, and may be used when $   \liminf_{|x| \rightarrow
\infty} {\frac {|\nabla v(x)|} {|\nabla h(x)|}}> 0.$ Write
\[
\int \varsigma^2  \nabla \phi \cdot \nabla \phi  = {\underbar c}
\int \nabla \phi \cdot \nabla \phi + \int (\varsigma^2 -
{\underbar {\sf c}})  \nabla \phi \cdot \nabla \phi.
\]
We now deduce a convenient lower bound on:
\[
\int (\varsigma^2 - {\underbar {\sf c}})  \nabla \phi \cdot \nabla
\phi,
\]
following an approach of \citet{davies89} (see Theorem 1.5.12).
Construct an additional potential function:
\[
 {\check W}(x) = (\varsigma^2 + {\underbar {\sf c}})(\nabla v \cdot
\nabla v) + \left(\varsigma^2 - {\underbar {\sf c}}\right) {\rm
trace}\left({\frac {\partial^2 v} {\partial x_i
\partial x_j}}\right).
\]

\begin{lemma} \label{lemma:davies}  If equation (\ref{cpotentiald})
holds, then:
\[
\int {\check W} \phi^2 \le  \int (\varsigma^2 - {\underbar {\sf c}})
\nabla \phi \cdot \nabla \phi ~~for~all~~ \phi \in C_K^2 .
\]
\end{lemma}

Note that
\begin{eqnarray*}
{\check V}(x) + {\check W}(x ) & = & \varsigma(x)^2 {\rm
trace}\left[{\frac {\partial^2 v(x)} {\partial x_i
\partial x_j}}- {\frac {\partial^2 h(x)} {\partial x_i \partial
x_j}}\right] + \varsigma(x)^2| \nabla h(x) - \nabla v(x)|^2 \cr &
& +{\underbar c}\left[\nabla v(x) \cdot \nabla v(x) -  {\rm
trace}\left({\frac {\partial^2 v(x)} {\partial x_i
\partial x_j}}\right)\right].
\end{eqnarray*}

\begin{criterion} \label{crit3} Equation (\ref{cpotentiald}) is satisfied for $\varsigma(x) = \exp[v(x)]$ and,

\begin{description}

\item[a)]
\[
\lim_{|x| \rightarrow \infty} \left[\nabla v(x) \cdot \nabla v(x)
- {\rm trace}\left({\frac {\partial^2 v(x)} {\partial x_i
\partial x_j}}\right)\right] = 0;
\]
\item[b)]
\[
\lim_{|x| \rightarrow \infty} \varsigma(x)^2 {\rm
trace}\left[{\frac {\partial^2 v(x)} {\partial x_i
\partial x_j}}- {\frac {\partial^2 h(x)} {\partial x_i \partial
x_j}}\right] + \varsigma(x)^2| \nabla h(x) - \nabla v(x)|^2 = +
\infty.
\]
\end{description}
\end{criterion}

\begin{proposition} Suppose Assumptions \ref{adensitys} and Condition \ref{cond:core} are satisfied.
Then Criterion \ref{crit3} implies Condition \ref{condhigh}.
\end{proposition}

\noindent Restriction a) of Criterion \ref{crit3} limits the tail
growth of the penalization.  There are two reasons that such growth
should be limited.  The fast growth in $\Sigma$ limits the functions
that we hope to approximate using NPCs.  Also for $C_K^2$ to be a
core for the form $f$ we require limits on growth in $\Sigma$ (see
subsection \ref{subsect:core}.)

Our use of ${\check W}$ in addition to ${\check V}$ in effect
replaces $- \varsigma^2 | \nabla v|^2$ with a second derivative
term:
\[
\varsigma(x)^2 {\rm trace}\left[{\frac {\partial^2 v(x)} {\partial
x_i
\partial x_j}}\right].
\]
The following example illustrates the advantage of this
replacement.

\begin{example}
\label{compacte2} Let
\[
v(x)  =  {\frac \beta 2} \log(1 + |x|^2) + {\frac {\tilde {\sf c}}
2}
\] where ${\tilde {\sf c}} = \log {\underbar {\sf c}}$. Thus
$\varsigma$ grows like $|x|^\beta$ in the tails. Simple
calculations result in
\[
- \nabla v(x) \cdot \nabla v(x) = - \beta^2 {\frac {|x|^2}{(1 +
|x|^2)^2}},
\]
and
\[
{\rm trace} \left[{\frac {\partial^2 v(x)} {\partial x_i
\partial x_j}}\right] =  \beta \left[{\frac {n + (n-2)|x|^2}{(1 +
|x|^2)^2}}\right].
\]
Notice that both terms converge to zero as $|x|$ gets large, but
that the squared gradient scaled by $\varsigma^2$ becomes
arbitrarily large when $\beta
> 1$. The first term is always negative, but the second one is
nonnegative provided that $n \ge 2$.  Even when $n=1$ the second
term is larger than the first provided that $\beta >
1$.\footnote{We have previously established an alternative
compactness criterion for $n = 1$ that does not involve second
derivatives that may be preferred to Criterion \ref{crit3}.} This
example illustrates when Criterion \ref{crit3} is preferred to
Criterion \ref{compact1}.  The distinction can be important when
densities have algebraic tails.
\end{example}

This section contains our main existence results, which we now
summarize. We provided two criteria for constructing penalization
functions that support the existence of countable many NPCs. The
first one, Criterion \ref{compact1} gives the most flexibility in
terms of the penalization matrix $\Sigma$; but it is applicable for
densities that have relatively thin tails. Densities with algebraic
tails are precluded.  The second one, Criterion \ref{crit3}, allows
for densities with algebraic tails but requires that the
penalization be more severe in the extremes to compensate for the
tail thickness. Making the penalization more potent limits the class
of functions that are approximated. Moreover, when the penalization
is too extreme, we encounter an additional approximation problem:
the family of functions $C_K^2$ ceases to be a core for the form
used in the NPCs extraction.
%

\section{Forms and Markov processes} \label{sec:forms}

So far we considered the role of the penalization matrix $\Sigma$ in
the construction and approximation properties of NPCs.  We now use
stationary Markov diffusions  to give an explicit interpretation of
this penalization matrix.

We proceed as follows.  Suppose the data $\{ x_i \}_{i=1}^{T}$ are
generated by a Markov diffusion by sampling say at integer points in
time. Specifically, $\{ x_t :t\ge 0 \}$  solves
$$ dx_{t}=\mu (x_{t})dt+\Lambda
(x_{t})dB_{t}
$$ for some n-dimensional vector function $\mu$ and some n by n matrix function $\Lambda$
of the Markov state with appropriate boundary restrictions, where $\{B_{t}:t\geq 0\}$
is an n-dimensional, standard Brownian motion.  Suppose further that this process has $q$
as its stationary density and that $\Sigma = \Lambda \Lambda'$.  We will have more to say
in Section \ref{sec:irrev} about the restrictions on $\mu$ that are implicit in such a construction.  Let
$\phi$ be in $C_K^2$.  Then it follows from Ito's lemma that the local variance of the process
$\{ \phi(x_t) \}$ is
\[
\left(\nabla \phi\right)' \Sigma \left(\nabla \phi \right)
\]
which is state dependent.  Note that $f(\phi,\phi)$ is the average of this local variance.  The local
variance is measure of magnitude of the instantaneous forecast error in forecasting $\{\phi(x_t)\}$ over the next instant given the current Markov state.

The NPC extraction given by Definition \ref{d:sprina} can be performed
equivalently as:

\begin{definition}
The function $\psi_{j}$ is the $j^{th}$ nonlinear principal
component (NPC) for $j\geq 0$ if $\psi _{j}$ solves:
\[ \min_{\phi \in \bar{H}} f(\phi, \phi)
\]
subject to
\begin{eqnarray*}
<\phi,\phi> & = & 1 , \\
<\psi_s,\phi> & = & 0, s=0,...,j-1.
\end{eqnarray*}
\end{definition}

\noindent Thus the NPCs are extracted by making the local forecast
error (appropriately scaled) small for functions with unit second
moments plus orthogonality. It is a continuous-time counterpart to
the $( 1- R^2 )$ in a forecasting regression. Recall that the NPCs
satisfy $< \psi_j, \psi_k> = 0$ and $f(\psi_j, \psi_k) = 0$ for $j
\ne k$. These properties are nonlinear counterparts to the canonical
components in the extraction of \citet{boxtiao}. \citet{boxtiao} show
that their canonical analysis produces $k$ component series that i)
are ordered from least predictable to most predictable, ii) are
contemporaneously uncorrelated, and iii) have contemporaneously
uncorrelated forecast errors.  In verifying our counterpart to the
third property, notice that in continuous time the unpredictable
component is $\nabla \psi_j(x_t) \Lambda(x_t) dB_t$, and thus
$f(\psi_j, \psi_k)$ is the (average) local covariance of
$\psi_j(x_t)$ and $\psi_k(x_t)$.

For financial and economics applications it is important to allow
for barriers that are not attracting, and it is desirable to allow
for a non-compact state space of the Markov process. Thus imposing
uniform bounds on both $q$ and the matrix $\Sigma$ over compact
state spaces is too restrictive. Our existence results in Section
\ref{sec:exist} avoid such restrictions.

Our construction of NPCs supports the estimation and testing of
multivariate Markov diffusion models. There are other functional
principal components constructions. For instance,
\citet{dauxoisnkiet} construct nonlinear principal components for
multivariate densities by choosing pairs of functions that maximize
cross correlations without penalizing derivatives. \citet{zh} propose
$L^1$-norm constrained principal components for the purpose of
dimension reduction and variable filtering.  \citet{ramsaysilverman}
provide detailed discussions on functional principal component
analysis for IID realization of curves.

\section{Reversible diffusions} \label{sec:reverse}

We next consider how to use the form $f$ to build a Markov process.
Specifically associated with the form $f$, there is a second-order
differential operator $F$ that generates the semigroup of a Markov
diffusion.  The diffusion process has $\Sigma$ as its local
covariance matrix and $q$ as it stationary density.  The
construction of $F$ is \emph{unique} provided that we restrict the
process to be time reversible.

\subsection{A differential operator} \label{subsec:differ}

There is a differential operator $F_o$ that is associated with the
form $f_o$ (given in \ref{edefpar}), which we construct using
integration-by-parts. For any functions $\phi$ and $\psi$ in
$C_K^2$:
\begin{eqnarray}
\nonumber f_o(\phi,\psi) & = & {\frac 1 2} \int \sum_{i,j}
\sigma_{ij} {\frac {\partial \phi} {\partial y_j}}{\frac {\partial
\psi} {\partial
y_i}}q \\
& = & - {\frac 1 2} \int \sum_{i,j} \sigma_{ij} {\frac {\partial^2
\phi} {\partial y_i \partial y_j}} \psi q - {\frac 1 2} \int
\sum_{i,j} {\frac {\partial (q \sigma_{ij})} {\partial y_i}}
{\frac {\partial \phi} {\partial y_j}} \psi \label{esymmetric}
\end{eqnarray}
where the second equality of (\ref{esymmetric}) follows from the
integration-by-parts formula:
\begin{equation*}
\int \sum_{i,j} {\frac {\partial (q \sigma_{ij})} {\partial y_i}}
{\frac {\partial \phi} {\partial y_j}} \psi = - \int \sum_{i,j}
\sigma_{ij} {\frac {\partial^2 \phi} {\partial y_i \partial y_j}}
\psi q - \int \sum_{i,j} \sigma_{ij} {\frac {\partial \phi}
{\partial y_j}}{\frac {\partial \psi} {\partial y_i}}q.
\end{equation*}
We use (\ref{esymmetric}) to motivate our interest in the
differential operator $F_o$:
\begin{equation}
\label{param} F_o\phi = - {\frac 1 2} \sum_{i,j} \sigma_{ij}
{\frac {\partial^2 \phi} {\partial y_i \partial y_j}} -  {\frac 1
{2q}} \sum_{i,j} {\frac {\partial (q \sigma_{ij})} {\partial y_i}}
{\frac {\partial \phi} {\partial y_j}}.
\end{equation}
This operator is constructed so that the form $f_o$ can be
represented as:
\begin{equation*}
f_o(\phi,\psi)  =  <F_o \phi, \psi>  =  <\phi , F_o \psi>,
\end{equation*}
where the second relation
holds because we can interchange the role of $\phi$ and $\psi$ in
(\ref{esymmetric}).  Notice from (\ref{param}) that operator $F_o$
has both a first derivative term and a second derivative term.
Symmetry (with respect to $q$) is built into the construction of
this operator because of its link to the symmetric form $f_o$.

We are  interested in the operator $F_o$  because of its use in
modeling Markov diffusions. Suppose that $\{x_t : t\ge 0 \}$ solves
the stochastic differential equation: $$ dx_{t}=\mu
(x_{t})dt+\Lambda (x_{t})dB_{t}
$$ with appropriate boundary restrictions, where $\{B_{t}:t\geq 0\}$
is an n-dimensional, standard Brownian motion, and:
\begin{equation*}
\mu_j ={\frac{1}{2q}}\sum_{i=1}^n{\frac{\partial
(\sigma_{ij}q)}{\partial y_{i}}}.
\end{equation*}
Set
\[
\Sigma=\Lambda \Lambda'.
\]
Then we may use Ito's Lemma to show that for each $\phi \in C_K^2$
\[
- F_o \phi=\lim_{t\downarrow 0} {\frac {E \left[ \phi (x_t)|x_0=x
\right] - \phi(x)} {t}},
\]
where this limit is taken with respect to the $L^2$. That is, $-F_o$
coincides with the \textit{infinitesimal generator} of $\{x_t\}$ in
$C_K^2.$  We use this link to the stochastic differential equation
to motivate our use of the matrix $\Sigma$ for penalizing
derivatives.  This matrix will also be the diffusion matrix for a
continuous-time Markov process with stationary density $q$.

\subsection{Generating reversible diffusions}
\label{subsec:reversible}

\citet{wong} constructed scalar diffusion models with stationary
densities  in the Pearson class and  produced a spectral or
eigenfunction decomposition of the associated one-parameter
semigroup of conditional expectation operators. \citet{banon} and
\citet{ckc} extended this analysis in part by taking as given the
diffusion coefficient and constructing the implied drift coefficient
for the stochastic differential equation that generates a prescribed
stationary density. \citet{banon} and \citet{ckc} did not analyze the
implied spectral decomposition of the associated conditional
expectation operators. In all  these analyses, the stationary
density of the diffusion process is taken as one of the starting
points of a model builder. In this section we share \citet{banon}'s
aim for generality, but at the same time we retain \citet{wong}'s
interest in spectral decompositions.

As in \citet{wong,banon,ckc}, we parameterize diffusion processes
using the stationary density $q$ and a (possibly state dependent)
diffusion coefficient $\Sigma$ in contrast to the more typical
approach of starting with a drift and the diffusion coefficients.
In contrast to \citet{wong,banon,ckc}, we allow the diffusion
process to be multivariate on a state space $\Omega.$   For this to result in
a unique diffusion, we require that the diffusion be time reversible.

A stochastic process is time reversible if its forward and
backward transition probabilities are the same.   Multivariate
reversible diffusions can be parameterized directly by the pair
$(q,\Sigma)$. Associated with the closed extension $f$ is a family
of resolvent operators $G_\alpha$ indexed by a positive parameter
$\alpha$. We use the resolvent operators to build a semigroup of
conditional expectation operators for a Markov process, and in
particular, the generator of that semigroup.

For any $\alpha>0$, the resolvent operator $G_\alpha$ is
constructed as follows. Given a function $\phi \in L^2$, define
$G_\alpha \phi \in {\mathcal D}(f)$ to be the solution to
\begin{equation}
\label{resolve} f(G_\alpha \phi,\psi) + \alpha<G_\alpha \phi,\psi>
= <\phi,\psi>
\end{equation}
for all $\psi \in {\cal D}(f)$.  The Riesz Representation Theorem
guarantees the existence of the $G_\alpha \phi$.  This family of
resolvent operators is known to satisfy several convenient
restrictions ({\it e.g.} see \citet{fot} pages 15 and 19). In
particular, $G_\alpha$ is a one-to-one mapping from $L^2$ into
$G_\alpha(L^2)$.

We associate with the form $f$ the self-adjoint, positive
semidefinite operator:
\begin{equation}
\label{operator} F \phi = (G_\alpha)^{-1} \phi - \alpha \phi
\end{equation}
defined on the domain $G_\alpha(L^2)$. It can be shown that $F$ is
independent of $\alpha$.  Since the operator $F$ is
self-adjoint and positive semidefinite, we may define a unique
positive semidefinite square root $\sqrt{F}$.   While $F$ may only
be defined on a reduced domain, the domain of its square root may be
extended uniquely to the entire space ${\mathcal D}(f)$ and: $
f(\phi,\psi) = <\sqrt{F} \phi, \sqrt{F} \psi> $ ({\it e.g.} see
\citet{fot} Theorem 1.3.1).  Moreover, it is an extension of the
operator $F_o$ because $f$ is an extension of $f_o$ ({\it e.g.} see
Lemma 3.3.1 of \citet{fot}).

We also use the family of resolvent operators to build a semigroup
of conditional expectation operators. A natural candidate for this
semigroup is $\{ \exp(-tF) : t \ge 0 \}$. Formally, the expression
$\exp(-tF)$ is not well defined as a series expansion. However,
for any $\alpha$ and any $t$, we may form the exponential:
\[
\exp(t\alpha^2 G_\alpha - \alpha t I)
\]
as a Neumann series expansion. Notice that (\ref{operator})
implies
\begin{equation*}
t \alpha^2 G_\alpha - t \alpha I  =  t\alpha[(I + {\frac 1 \alpha}
F)^{-1} - I]  =  -tF \left( I + {\frac 1 \alpha} F \right)^{-1}.
\end{equation*}
Instead directly using a series expansion, we use the limit
\[
\lim_{\alpha \rightarrow \infty}\exp[(t\alpha^2 G_\alpha) - \alpha
t I]= \exp(-tF)
\]
often referred to as Yosida approximation to construct formally a
strongly continuous, semigroup of operators indexed by $t \ge
0$.

We have just seen how to construct resolvent operators  and the
semigroup of conditional expectation operators from the form.  We
may {\it invert} this latter relation and obtain:
\begin{equation}
G_\alpha \phi = \int_0^\infty \exp(-\alpha t) \exp(-tF)\phi dt
\label{resolve1}
\end{equation}
which is  the usual formula for the resolvents of a semigroup of
operators. The operator $-F$ is referred to as the generator of
both the semigroup $\{ \exp(-tF) : t \ge 0\}$ and of the family of
resolvent operators $\{ G_\alpha : \alpha > 0 \}$.

As we have just seen, associated with a closed form $f$, there is an
operator $F$ and a (strongly continuous) semigroup $\{ \exp(-tF) : t
\ge 0\}$ on $L^2$.  To establish that there is a Markov process
associated with this semigroup, we need first to verify that the
semigroup satisfies two properties.  First we require, for each $t
\ge 0$ and each $0 \le \phi \le 1$ in $L^2$, $0 \le \exp(-tF)\phi
\le 1$. A semigroup satisfying this property is called {\it
submarkov} in the language of \citet{beurlingdeny}. Second we
require, for each $t \ge 0$, $\exp(-t F)1 = 1$. A semigroup
satisfying this property is said to {\it conserve probabilities}. We
refer to a submarkov semigroup that conserves probabilities as a
Markov semigroup. Finally we must make sure that the
Markov semigroup is actually the family of conditional expectation
operators of a Markov process.

The following condition is sufficient for a closed form to
 generate a submarkov semigroup (\textit{e.g.}, see \citet{davies89} section
 1.3).
\begin{condition} (Beurling-Deny) \label{beurlingdeny}
For any $\phi \in {\cal D}(f)$, $\psi$ given by the truncation:
\[
\psi = (0 \vee \phi) \wedge 1
\]
is in ${\cal D}(f)$ and
\[
f(\psi, \psi) \le f(\phi, \phi).
\]
\end{condition}
\noindent When this condition is satisfied, the semigroup
$\exp(-tF)$ is submarkov, and for each $t \ge 0$, $\exp(-t F)$ is an
$L^2$ contraction ($\|\exp(-t F ) \phi \|_2 \le \| \phi \|_2$). This
contraction property is also satisfied for the $L^p$ norm for $1 \le
p \le \infty$ (\citet{davies89} Theorem 1.3.3).
  In particular, we may extend the semigroup from
$L^2$ to $L^1$ while preserving the contraction property.

\begin{proposition}\label{rev-markov} There exists a self adjoint operator $F$
associated with $f$, which is an extension of $F_o$ and generates a
semigroup $\{ \exp(- tF) : t \ge 0 \}$.  The density $q$ is the
stationary density for this diffusion, the matrix $\Sigma$ is the
diffusion matrix and $\exp(-tF)$ is the conditional expectation
operator over an interval of time $t$.
\end{proposition}

%

\subsection{Nonlinear principal components and eigenfunctions}

Continuous time Markov process models are typically specified in
terms of their local dynamics.  Given the nonlinearity in the state
variables, it is a nontrivial task to infer the global dynamics, and
in particular the long-run behavior from this local specification.
Characterizing eigenfunctions  of conditional expectation operators
offer a way of approximating intermediate and long term dynamics in
ways that are typically disguised from the local dynamics in
nonlinear settings.

Eigenfunctions of the closed form $f$ will also be eigenfunctions
of the resolvent operators $G_\alpha$ and of the generator $F$.
For convenience, we rewrite equation (\ref{resolve}):
\begin{equation*}
f(G_\alpha \phi,\psi) + \alpha<G_\alpha \phi,\psi> = <\phi,\psi>.
\end{equation*}
From this formula, we may verify that $f$ and $G_\alpha$ must share
eigenfunctions for any $\alpha
> 0$. The eigenvalues are related via the formula:
\begin{equation*}
\lambda = {\frac 1 {\delta + \alpha}}
\end{equation*}
where $\lambda$ is the eigenvalue of $G_\alpha$ and $\delta$ is
the corresponding eigenvalue of $f$.

Given the relation between the generator $F$ and the resolvent
operator $G_\alpha$,
\[
F\phi  = (G_\alpha)^{-1}\phi  - \alpha \phi,
\]
these two operators must share eigenfunctions. Moreover,
eigenfunctions of the operators $F$, $G_\alpha$ and the form $f$
must belong to the domain of $F$ or equivalently to the image of
$G_\alpha$. This domain is contained in the domain of the form $f$.
Similarly, we may show that if $\phi$ is an eigenfunction of the
form $f$ with eigenvalue $\delta$, then $\phi$ is an eigenfunction
of $\exp(-tF)$ with eigenvalue $\exp(-t \delta)$ for any positive
$t$.

An eigenfunction $\psi$ of the generator $F$ satisfies:
\begin{equation}
\label{tmoment} E[\psi(x_{t+s}) | x_t] = \exp(-\delta s)
\psi(x_t),
\end{equation}
for some positive number $\delta$ and each transition interval $s$.
Thus the NPCs described previously will also satisfy the testable
conditional moment restrictions (\ref{tmoment}). The scalar process
$\{ \psi(x_t) \}$ should behave as a scalar autoregression with
autoregressive coefficient $\exp(-\delta s)$ for sample interval
$s$.  The forecast error: $\psi(x_{t+s}) - \exp(-\delta s)
\psi(x_t)$ will typically be conditionally heteroskedastic (have
conditional variance that depends on the Markov state $x_t$).

Since the form $f$ can be depicted using a principal component
decomposition as in (\ref{eigen:fdecom}), analogous decompositions
are applicable to $F$ and $\exp(-tF)$:
\begin{eqnarray*}
F\phi & = & \sum_j \delta_j{\frac { < \phi, \psi_j>}{ <\psi_j,
\psi_j>}} \psi_j , \cr \exp(-tF)\phi & = & \sum_j \exp( - t
\delta_j)
 {\frac {< \phi, \psi_j>}{<\psi_j, \psi_j>}} \psi_j ,
\end{eqnarray*}
where the first expansion is only a valid $L^2$ series when $\phi$
is in the domain of the operator $F$. When the eigenvalues
$\delta_j$ of the form increase rapidly (in $j$), the term $\exp( -
t \delta_j)$ will decline to zero rapidly (in $j$), more so when the
time horizon $t$ becomes large. As a consequence, it becomes easier
to approximate the conditional expectation operator over a finite
transition interval $t$ with a smaller number of NPCs.  On the other
hand, slow eigenvalue divergence of the form will make it
challenging to approximate the transition operators with a small
number of NPCs. Our results in Section \ref{sec:exist} give
primitive conditions based on the tail behaviors of stationary
density and diffusion matrix for the existence of the above
eigenfunction decompositions. In an earlier longer version of our
paper we provided primitive sufficient conditions for the speed of
eigenvalue decays.

\subsection{An alternative form}

In this subsection we construct a second quadratic form used to
depict the long-run variance of a stochastic processes constructed
from the Markov process $\{x_t\}$.

This quadratic form is defined to be the limit
\[
g(\phi,\psi) =2 \lim_{\alpha \downarrow 0} <G_\alpha \phi, \psi>
\]
and is well defined on a subspace ${\mathcal S}(F)$ of functions
in $L^2$ for which
\[
 \lim_{\alpha \downarrow 0} <G_\alpha \phi, \phi> < \infty.
\]
While the form $f$ is used to define the operator $F$, the form
$g$ may be used to define $F^{-1}$ as is evident from formulas
(\ref{resolve}) or (\ref{operator}).  The forms $f$ and $g$ share
eigenfunctions.  The $g$ eigenvalues are the reciprocals of the
$f$ eigenvalues.

In light of equation (\ref{resolve1})
\begin{equation}
\label{resolvelim} <G_\alpha \phi, \psi > = \int_0^\infty
\exp(-\alpha t) E[\phi(x_t) \psi(x_0)]dt.
\end{equation}
Hence, using (\ref{operator}), we obtain:
\begin{equation*}
g(\phi, \psi)  =  \lim_{\alpha \downarrow 0} 2 <G_\alpha \phi, \psi>
 =  \lim_{\alpha \downarrow 0} 2 <(\alpha I +F)^{-1} \phi, \psi>.
\end{equation*}
Notice that this form is symmetric because the resolvent operator is
self-adjoint for any positive $\alpha$.  Using (\ref{resolvelim}) we
may write this form as
\begin{equation*}
g(\phi,\psi)  =  \int_{-\infty}^{+\infty} E[\phi(x_t) \psi(x_0)]dt =
 \int_{-\infty}^{+\infty} E[\psi(x_t) \phi(x_0)]dt.
\end{equation*}

\begin{proposition} The $j^{th}$ nonlinear principal component $\psi_{j}$
for $j\geq 1$ solves:
\[ \max_{\phi \in {\mathcal S}(F)} g(\phi, \phi)
\]
subject to
\begin{eqnarray*}
<\phi, \phi> & = & 1 , \\
<\psi_s,\phi> & = & 0, s=0,...,j-1 ,
\end{eqnarray*}
where $\psi_{0}$ is initialized to be the constant function one.
\end{proposition}

Recall that the spectral density function at frequency $\theta$
for a stochastic process $\{\phi(x_t)\}$ is defined to be:
\[
\int_{-\infty}^{+\infty} \exp(-i \theta t) E[\phi(x_t) \phi(x_0)]
dt
\]
whenever this integral is well defined.  In particular
$g(\phi,\phi)$ is the spectral density of the process
$\{\phi(x_t)\}$ at frequency zero, a well known measure of the
long-run variance.

For an alternative but closely related defense of the term {\it long-run}
 variance, suppose that $\phi = F\psi$ for some $\psi$ in the
domain of $F$.  Then,
\[
M_T = \psi(x_T) - \psi(x_0) + \int_0^T \phi(x_s) ds
\]
is a martingale adapted to the Markov filtration. Following
\citet{bhattacharya} and \citet{hsback}, we may use this martingale
construction to justify: $$
{\frac{1}{\sqrt{T}}}\int_{0}^{T}\phi (x_{s})ds\Rightarrow \mathrm{Normal}%
\;\left(0,g(\phi,\phi)\right). $$ Thus $g(\phi,\phi)$ is the
limiting variance for the process
$\{{\frac{1}{\sqrt{T}}}\int_{0}^{T} \phi (x_{s})ds\}$ as the sample
length $T$ goes to infinity.

This gives us an alternative interpretation of our NPCs. We may base
the extraction on maximizing $g(\phi,\phi)$ subject to $<\phi,\phi>
= 1$ and orthogonality constraint. In words we are maximizing
long-run variation while constraining the overall variation.  Smooth
functions of a Markov state are also highly persistent and as a
consequence maximize long-run variation.

\section{Irreversible diffusions}\label{sec:irrev}

The stationary Markov construction we used in the previous section
resulted in a generator that was self adjoint and hence a process
that was time reversible.  Even among the class of stationary Markov
diffusions, reversibility is special when the process has multiple
dimensions.  Given a stationary density $q$ and a diffusion matrix
$\Sigma$, we have seen how to  construct a reversible diffusion, but
typically there are other diffusions that share the same density and
diffusion matrix, but not reversible. We now characterize the drifts
of such processes.

Instead of constructing a Markov process implied by a form, suppose
we have specified the process as a semigroup of conditional
expectation operators indexed by the transition interval. We suppose
this process has stationary density $q$. Following \citet{nelson} and
\citet{hsback} we study the semigroup of conditional expectation
operators on the space $L^2$.  This semigroup has a generator $A$
defined on a dense subspace of $L^2$. Consistent with our
construction of $F$, on the subspace of $C_K^2$, we suppose that $A$
can be represented as a second-order differential operator:
\[
A \phi = {\frac 1 2} \sum_{i,j} \sigma_{ij} {\frac {\partial^2 \phi}
{\partial y_i \partial y_j}} + \sum_{j} \mu_j {\frac {\partial \phi}
{\partial y_j}},
\]
and that
\[
\int A \phi q = 0.
\]
It may be shown that
\[
-\int \psi (A \phi) q = f_o(\phi,\psi)~~on~~ C_K^2 .
\]

This construction does not require that $A = -F$ or that $A$ be
self adjoint.  How can the adjoint be represented?  The adjoint
must satisfy:
\[
-\int \phi (A^* \psi) q = f_o(\phi,\psi),
\]
implying that the $F$ constructed previously must satisfy: $F = -(A
+ A^*)/2$.  Moreover, since $q$ is also the stationary density of
the reverse time process:
\[
\int A^* \phi q = 0.
\]

It follows from \citet{nelson} that the adjoint operator has the same
diffusion matrix, but a different drift vector. The drift for the
adjoint operator $A^*$ is given by:
\[
\mu^* = - \mu + {\frac 1 {q}} \sum_{i,j} {\frac {\partial (q
\sigma_{ij})} {\partial y_i}} {\frac {\partial \phi} {\partial
y_j}}.
\]
The adjoint operator generates the semigroup of expectation
operators for the reverse time diffusion.  From the formula for
reverse time drift, $\mu^*$, it follows that
\[
{\frac {\mu + \mu^*} 2} = {\frac 1 {2q}} \sum_{i,j} {\frac
{\partial (q \sigma_{ij})} {\partial y_i}} {\frac {\partial \phi}
{\partial y_j}},
\]
which is the negative of the second term in representation
(\ref{param}) for $F_o$. Thus if the generator $A$ of the
semigroup is not self adjoint, then the operator $F$ implied by
the form is a second order differential operator built using a
simple average of the forward and reverse time drift coefficients,
$\mu$ and $\mu^*$, and the common diffusion matrix, $\Sigma$.

\begin{remark}
The density $q$ and the diffusion matrix $\Sigma$ do place other
restrictions on the drift vector $\mu$.  Since $q$ is the
stationary density, $\mu$ and $\mu^*$ must also satisfy:
\begin{eqnarray*}
{\frac {\partial (\mu q) } {\partial y}} & = & {\frac {\partial}
{\partial y}} \sum_{i,j} {\frac {\partial (q \sigma_{ij})}
{\partial y_i}} {\frac {\partial \phi} {\partial y_j}}, \cr {\frac
{\partial (\mu^* q) } {\partial y}} & = & {\frac {\partial}
{\partial y}} \sum_{i,j} {\frac {\partial (q \sigma_{ij})}
{\partial y_i}} {\frac {\partial \phi} {\partial y_j}}.
\end{eqnarray*}
While there is typically one solution $\mu$ (or $\mu^*$) to this
equation for the scalar case, multiple solutions will exist for the
multivariate case.  That is, unless reversibility is imposed {\it a
priori}, the drift cannot be identified from the density and
diffusion matrix; but the average of the forward and backward drift
can be inferred.
\end{remark}

\begin{remark}

The NPCs existence results of section \ref{sec:exist} have an
immediate extension to the existence of eigenfunctions of the
semigroup of conditional expectation operators when the Markov
diffusion is not reversible. For a semigroup with generator $A$ we
may ``invert'' equation \ref{operator} to construct a family of
resolvent operators:
\[
R_\alpha \phi = \int_0^\infty \exp(-\alpha t) \exp(At) \phi dt =
(\alpha I - A)^{-1}\phi ,
\]
and a  form $f(\phi,\psi) = <\phi,A \psi>$, which is not
necessarily symmetric.  While the generator is an unbounded
operator on $L^2$, the resolvent operators are bounded.  When the
resolvents are compact operators, they have well defined
eigenfunctions and eigenvalues, but they may be complex valued.
(See \citet{rudin}, Theorem 4.25, page 108.)

Given $\alpha$ the resolvent operator will be compact provided
that the image of $R_\alpha$ of the $L^2$ unit ball has compact
closure. Consider a function $\varphi$ given by
\[
\varphi = (\alpha I - A)^{-1} \phi.
\]
Then $\phi \in {\mathcal D}(A)$ and
\[
<\phi , \phi> = \alpha^2 <\varphi, \varphi> - 2\alpha <\varphi, A
\varphi> + <A \varphi, A \varphi> \ge \alpha^2 <\varphi, \varphi> +
2\alpha f(\varphi,\varphi).
\]
Thus it suffices to show that
\[
\{ \varphi \in {\mathcal D}(A) : \alpha^2<\varphi,\varphi> +
2\alpha f(\varphi, \varphi) \le 1\}
\]
has compact closure.  This set will have compact closure, if, and
only if, compactness Condition \ref{condhigh} is satisfied for
$\theta = \alpha /2$.
\end{remark}

\section{Extraction of NPCs from data} \label{sec:consist}

In this section we suggest a sieve method to
estimate NPCs based on discretely-sampled data from a Markov process.
Here we only sketch the construction and
leave the formal justification and detailed analysis of rates of convergence
for subsequent research.

%

Let $\{x_i\}_{i=1}^{T}$ be a discrete-time sample of
the underlying continuous-time, ergodic Markov diffusion $\{x_t:t\geq
0\}$ on the state space $\Omega \subseteq \mathbb{R}^n$.\footnote{If available a continuous-time
record could be used, but statistical approximation remains an issue because the
length of the record is finite.}
Suppose that the penalization matrix $\Sigma$ is either known or
could be consistently estimated. The invariant probability measure
$Q$ is unknown but is consistently estimated by the empirical
distribution of the data $\{x_i\}_{i=1}^{T}$.

Let $\{H_{m} : m=1,2,...\}$ be a sequence of increasing
finite-dimensional linear (sieve) spaces that approximate the
Hilbert space $\bar H$ (the domain of the form $f$) as $m$ goes to
infinity. For notational convenience, let $m$ be the dimension of
$H_m$, and suppose that $m$ goes to infinity slowly as the sample
size $T$ goes to infinity. One strategy is to extract the
finite sample approximations sequentially as in optimization problem given in
Definition \ref{d:sprina}.  For a finite-dimensional sieve
approximation, it suffices to solve a generalized eigenvector problem.


Since the space $H_m$ is finite-dimensional,
\[
H_m =\left\{ \phi(x)=\sum_{k=1}^{m}b_{k}B_{k}(x) \right\},
\]%
where the basis functions $\{B_{k}(x):k\geq 1\}$ are used to construct the sieve.
For example, $\{B_{k}(x):k\geq 1\}$ could be one of
the following: (i)
thin-plate splines, radial-basis wavelets or tensor-product wavelets
if $q$ has algebraic tails on its support $\Omega = \mathbb{R}^n$;
(ii) tensor-product Hermite polynomial basis or Gaussian radial basis
if $q$ has exponential thin tails on its support $\Omega =
\mathbb{R}^n$.


Form a vector $\Psi_m$ of functions of $x$ by stacking terms
$B_{k}(x), ~k=1,...,m$, and form two matrices:
\begin{eqnarray*}
\textbf{V}_{T} &=&\frac{1}{2T}\sum_{i=1}^{T}\left[\nabla \Psi
_{T}(x_{i})\right]\Sigma
(x_{t})\left[\nabla \Psi _{T}(x_{t})\right]^{\ast } ,\\
\textbf{W}_{T} &=&\frac{1}{T}\sum_{i=1}^{T}\Psi _{T}(x_{i})\Psi
_{T}(x_{i})^{\ast } ,
\end{eqnarray*}%
where $^{\ast }$ denotes transpose. Both matrices are symmetric and
positive semidefinite by construction. The matrix $\textbf{W}_{T}$
is typically nonsingular while $\textbf{V}_T$ is singular when a
constant function is in the sieve space $H_m$. Stack the
coefficients on the sieve basis into a vector $a_{T}$. The sample
counterpart to the NPC (or eigenfunction) problem is the following
generalized eigenvector problem:

\begin{equation}
\textbf{V}_{T}\widehat{a}_{T}=\widehat{\delta
}_{T}\textbf{W}_{T}\widehat{a}_{T} \label{eig1}
\end{equation}%
where $\widehat{a}_{T}$ is a generalized eigenvector and $\widehat{\delta }%
_{T}$ a generalized eigenvalue. Since $\textbf{W}_{T}$ is positive
definite, we may apply the Cholesky decomposition to transform this
generalized eigenvector problem into a standard eigenvector problem.

Associated with each eigenvector solution to (\ref{eig1}), is an
eigenfunction formed by multiplying the coefficient entries of the
eigenvector by the sieve basis functions. We have constructed this
sample problem so that one of the approximating eigenfunctions will
be constant whenever there is a nonzero constant function in $H_m$,
and the associated eigenvalue is zero.\footnote{ For reversible
diffusions we could instead approximate  the NPCs nonparametrically
by maximizing autocorrelations. For scalar diffusion models, this
method has been already considered in \citet{chs98} and
\citet{markus}. Both papers used a wavelet sieve as the approximating
space.}

\section{Conclusions and related literature}\label{sec:conclude}

We have studied NPCs from multiple vantage points. We have explored
their role in capturing variation subject to smoothness constraints
and their role in capturing long-run variation in time series
modeling.  We have also considered their use in approximation where
the smoothness constraints limit the family of functions to be
approximated.

We also used multivariate Markov diffusions as data generating
devices to interpret our NPCs. These NPCs are eigenfunctions of
conditional expectation operators when the Markov process is
reversible and hence imply conditional moment restrictions. Our
analysis expands on the result of \citet{hsback} that reversible
diffusions can be identified nonparametrically from discrete-time
low frequency stationary observations.

For more general diffusions, these NPC are orthogonal and have
orthogonal innovations analogous to those from the canonical
analysis of \citet{boxtiao}  for linear multiple time series models.
(See also \noindent \citet{py}.)  Thus our NPC construction provides
a convenient way to summarize implications of multivariate nonlinear
diffusion models. Given the nonlinearity in the state variables, it
is a nontrivial task to infer the global dynamics, and in particular
the long-run behavior from this local specification based on low
frequency data. Our characterization of NPCs offers a way to
characterize features of the implied time series that are typically
disguised from the local dynamics. While we featured diffusion
processes, more general processes including processes with jumps can
be accommodated by expanding the types of forms that are considered.

The idea of using eigenfunctions of conditional expectation
operators for estimation and testing of Markov processes based on
low frequency data has been suggested previously by \citet{demoura},
\citet{hsback}, \citet{ks}, \citet{hst}, \citet{florensrenaulttouzi},
\citet{chs98} and \citet{markus}. In particular, \citet{ks} use
eigenfunctions to construct quasi-optimal estimators of parametric
scalar diffusion models of the drift and diffusion coefficients from
discrete-time data in the special case in which the functional forms
of eigenfunctions are known {\it a priori}. \citet{hsback},
\citet{hst}, \citet{chs98}, \citet{markus} and \citet{dfg} study
semiparametric and nonparametric identification and
over-identification based on an eigenfunction extraction that is
closely related to the one analyzed here; see \citet{fan} for a
recent review. This previous literature focuses primarily on scalar
diffusion models and in some cases to scalar diffusions on compact
state spaces with reflective boundaries. Our analysis of Markov
diffusions extends to multivariate settings applicable to processes
without attracting barriers.


In this paper we have characterized  a particular type of functional
principal components motivated in part by long-run implications of
multivariate Markov diffusions. This is a natural first step.
Inferential issues, while crucial, are beyond the scope of this
paper.
Formalizing statistical comparisons of models and data in a
multivariate setting is an obvious next step, supported by either
parametric, semiparametric or nonparametric estimation.
There are a number of  recent statistical results on estimation and
inference of functional principal components of covariance operators
associated with i.i.d. or longitudinal sample of curves. See {\em
e.g.}, \citet{silverman}, \citet{ramsaysilverman}, \citet{hmw},
\citet{bhk} and \citet{zhc}. These existing results can in principle
be modified to establish asymptotic properties of our estimated NPCs
from discrete-time low frequency realizations of an underlying
multivariate Markov diffusion model.

\newpage

\begin{appendix}

\section{Proofs}

\begin{proof}[Proof of Claim \ref{claim:bound}] In solving the maximization component of the
problem, first limit the $\phi$'s to be in $H_N$ but orthogonal to
${\hat H}$. This can only reduce maximized value. The space of such
$\phi$'s contains more than just the zero element because $H_N$ has
$N+1$ dimensions. Write $\phi$ as: $\phi = \sum_{j=0}^N {\sf r}_j
\psi_j$. Since $Proj(\phi|{\hat H}) = 0$, the objective can be
expressed as: $\sum_{j=0}^N ({\sf r}_j)^2\lambda_j $. The constraint
set implies that
\[
\sum_{j=0}^N ({\sf r}_j)^2 \le 1
\]
because $f(\psi_j, \psi_\ell) = <\psi_j, \psi_\ell> = 0$ for $j \ne
\ell$. While the coefficients ${\sf r}_j$ cannot be freely chosen
($\phi$ must be orthogonal to ${\hat H}$), they can be scaled so
that the constraint is satisfied with equality. Since the sequence
of $\lambda_j$'s is decreasing, the maximized objective must be no
less than $\lambda_N$.
\end{proof}

\begin{proof}[Proof of Claim \ref{claim:attain}] Write $\phi$
as: $\phi = Proj(\phi | H_{N-1}) + \varphi $ where $\varphi$ is in
$H_{N-1}^\perp$.  Write:
\[
Proj(\phi | H_{N-1}) = \sum_{j=0}^{N-1} {\sf r}_j \psi_j \] Using
this decomposition, the objective can be written as:
$<\varphi,\varphi>$, and the constraint set can be written as:
\[
\sum_{j=0}^{N-1} ({\sf r}_j)^2 + <\varphi,\varphi> + \theta
f(\varphi,\varphi) \le 1,
\]
because $\psi_1, \psi_2, ..., \psi_{N-1}, \varphi$ are orthogonal,
and $f(\psi_j,\varphi) = f(\psi_j,\psi_{\ell}) = 0$ for $j = 0,
...,N-1$ and $\ell = j+1, j+2,..., N-1$.  To maximize the objective,
the coefficients ${\sf r}_j$'s are set to zero and $\varphi$ is
chosen by solving Problem \ref{prob:jose1} for $H = H_{N-1}^\perp$.
The conclusion follows.
\end{proof}

\begin{proof}[Proof of Proposition \ref{scalar}] \citet{hst} consider
densities from stationary scalar diffusions, whose boundaries are
not attracting. This proposition gives an equivalent statement  of
their compactness condition, written in terms of the stationary
density.  The scalar diffusion coefficient in their analysis is
$\varsigma^2$.
\end{proof}

To show that the form $f$ is closed extension of $f_o$, we verify
that ${\bar H}$ is a Hilbert space.

\begin{proposition}
\label{hilbertspace} ${\bar H}$ is a Hilbert space.
\end{proposition}
\begin{proof}[Proof of Proposition \ref{hilbertspace}]
Let $\Lambda$ be the symmetric square root of the penalty matrix
$\Sigma $. If $\{\phi_{j}\}$ is a Cauchy sequence in ${\bar H}$,
then $\{\phi_{j}\} $ and the entries of $\{ \Lambda \nabla \phi
_{j}\}$ form  Cauchy sequences in ${L}^{2}$.  Denote the limits in $
{L}^{2} $ as
\begin{eqnarray*}
\phi &=&\lim_{j\rightarrow \infty }\phi _{j} \\
v &=&\lim_{j\rightarrow \infty }\Lambda \nabla \phi _{j}.
\end{eqnarray*}
For each $u\in {C}_{K}^{1}$ we know that: $$ \int \phi
_{j}\frac{\partial u}{\partial x}=-\int (\nabla \phi _{j})u,
$$ where $\frac{\partial u}{\partial x}$ is the partial derivative of $u$ with respect to $x$. Since $\Sigma $ is positive definite and continuous on any
compact subset of $\Omega $ and $u$ vanishes outside any such set,
it follows that $$ \int \phi \frac{\partial u}{\partial x}=-\int
(\Lambda ^{-1}v)u. $$ Hence $\phi \in {\bar H}$ with ${\nabla \phi
}=\Lambda^{-1}v$. Moreover, $\phi_{n}\rightarrow \phi $ in ${\bar
H}$.
\end{proof}


We now present a criteria for Condition \ref{cond:core} to hold.
This result is due essentially to \citet{azencott} and
\citet{davies85}.

\begin{proposition}\label{Markov} Consider a form $f_o$ that satisfies the
Beurling-Deny Criterion \ref{beurlingdeny}. Let ${\hat f}$ denote
the minimal extension of $f_o$ with domain  ${\mathcal D}({\hat
f}).$ Suppose that $1 \in {\mathcal D}({\hat f})$ and ${\hat
f}(1,\phi) = 0$ for all $\phi \in {\mathcal D}({\hat f})$. Then
${\hat f} = f$.
\end{proposition}

\begin{proof}[Proof of Proposition \ref{Markov}] As explained in Section \ref{sec:reverse}, associated with the forms $f$ and ${\hat
f}$ we may construct operators $F$ and ${\hat F}$ and resolvents $G$
and $\hat G.$ Integration by parts can be used to show that the
operators $F$ and ${\hat F}$ are extensions of the differential
operator \[ {\hat L} \phi = - {\frac 1 q} \sum_{i,j} {\frac \partial
{\partial x}_i} \left(q \sigma_{ij} {\frac {\partial \phi} {\partial
x_j}}\right),
\]
  defined on $C^2_K.$
The form ${\hat f}$ and hence the form $f$ satisfies the
Beurling-Deny Criterion \ref{beurlingdeny} (\citet{davies89} Theorem
1.3.5). Hence as stated in \citet{davies85} the operators $F$ and
${\hat F}$  can be extended to subspaces of $L^1.$ Similarly the
resolvents $G$ and $\hat G$ can be extended to $L^1.$ We will denote
the extended operators as $F^1,$ ${\hat F}^1,$ $G^1,$ and ${\hat
G}^1.$ Since $q$ is integrable, $L^2$ convergence implies $L^1$
convergence and consequently $F$ and ${\hat F}$ are restrictions of
$F^1 $ and ${\hat F}^1$, respectively. Similarly for the resolvent
operators.

If    ${\hat f}(1,\phi) = 0$ for all $\phi \in {\mathcal D}({\hat
f})$ then  ${\hat F}1=0$ and ${\hat G} 1=1.$ Consequently ${G}^ 1 1
=1.$ It follows from Theorem 2.2 in \citet{davies85} that $C_K^2$ is
a core for $F^1,$ in the sense that $F^1$ is the closure in $L^1$ of
${\hat L}$.\footnote{\citet{davies85} assumes that the coefficients
of ${\hat L}$ are $C^\infty.$ However the proof holds for $C^2$
coefficients since elliptic regularity holds even when the
coefficients are only Lipschitz (see Theorem 6.3 of \citet{agmon}) }
Hence $C_K^2$ is a core for $F^2,$ and thus a core for $f$ or,
equivalently, $f$ and ${\hat f}$ coincide.
\end{proof}



\begin{proof}[Proof of Proposition \ref{prop:recur}]  Since ${\hat f}$ is the minimal closed extension,
it has $C_K^2$ as its core. When this condition is met, a sequence
of functions $\phi_j$ in $C_K^2$ can be constructed that converge to
$1$ in $L^2$ and $f(\phi_j, \phi_j)$ converges to zero. See
\citet{fot} Theorem 1.6.6 and Theorem 1.6.7. An approximating
sequence of functions with compact support is supplied by \citet{fot}
in the proof of Theorem 1.6.7. This sequence can be smoothed using a
suitable regularization to produce  a corresponding approximating
sequence in $C_K^2$.  Thus the unit function is in the domain of
${\hat f}$ and ${\hat f}(1,\phi)=0$ for $\phi \in C_K^2$ and hence
for $\phi \in {\mathcal D}({\hat f})$. As we established above, this
is sufficient for  Condition \ref{cond:core}.
\end{proof}

\begin{proof}[Proof of Proposition \ref{prop:weaktrans}] Since $V$ is bounded from below, we may choose a $\theta > 0$ such
that $V + \theta$ is nonnegative. Construct the space:
\begin{eqnarray*}
{\check H} & \doteq &\{\psi \in L^2(leb): \int (V + \theta) (\psi)^2
< \infty, \mathrm{ there \ \ exists}\; g\; \mathrm{measurable, \ \
with}\; \int g^{\prime }\Sigma g < \infty,\\ & & \mathrm{and}\; \int
\psi \nabla \varphi = - \int g \varphi, \mathrm {for \ \ all }\;
\psi \in C_K^{1} \}.
\end{eqnarray*}
As in the proof of Proposition \ref{hilbertspace}, it follows that
${\check H}$ is a Hilbert space with inner product:
\[
\int (V + \theta+1) \psi {\tilde \psi} + \int (\nabla \psi)' \Sigma
(\nabla {\tilde \psi}).
\]
We show that $U {\bar H} \subset {\check H}$.

Since $C^2_K$ is a core for $f$, there exists a sequence $\{ \phi_j
: j=1,2,...\}$ in $C^2_K$ that converges to $\phi$ in the Hilbert
space norm of ${\bar H}$. Hence this sequence is Cauchy in that
norm. Writing $\psi_j=U \phi_j$ and applying equation (\ref{formn})
we obtain:
\begin{eqnarray*}
 \int (\phi_j - \phi_\ell)^2 (1 + \theta) q & + & \int (\nabla \phi_j - \nabla\phi_\ell)'
\Sigma  (\nabla \phi_j - \nabla \phi_\ell) q \cr & = & \int (V +
\theta +1) (\psi_j - \psi_\ell)^2 + \int (\nabla \psi_j - \nabla
\psi_\ell)'\Sigma (\nabla \psi_j - \nabla \psi_\ell).
\end{eqnarray*}
Thus $\{ \psi_j : j = 1,2,...\}$ is Cauchy in the Hilbert space norm
of ${\check H}$ and the limit point $\psi$ must satisfy $\psi = U
\phi$. Notice that $\int V (\psi)^2 + \int (\nabla \psi)' \Sigma
(\nabla \psi)$ equals ${\check H}$ squared norm minus $\theta +1$
times the $L^2(Q)$ squared norm. Thus,
\begin{eqnarray*}
\int V (\psi)^2 + \int (\nabla \psi)' \Sigma (\nabla \psi) & = &
\lim_{j \rightarrow \infty} \int (\nabla \psi_j)'\Sigma (\nabla
\psi_j) q \cr
 & = & \lim_{j \rightarrow \infty} (\nabla \phi_j)' \Sigma (\nabla
 \phi_j)q \cr
 & = & \int (\nabla \phi)' \Sigma (\nabla \phi) q.
 \end{eqnarray*}
This proves (\ref{weakbound}).

For a given $\psi = U \phi$ our candidate for the weak derivative
is,
\[
g \doteq \exp(-h)(-\phi \nabla h  + \nabla \phi ).
\]
To verify that $g$ is indeed the weak derivative, we must show that
for any $\varphi \in C^1_K$
\begin{equation} \label{firststep}
\int \psi \nabla \varphi = - \int g \varphi,
\end{equation}
and
\begin{equation} \label{secondstep}
\int g' \Sigma g < \infty.
\end{equation}
We check relation (\ref{firststep}) by applying integration by
parts,
\begin{eqnarray*}
- \int \nabla \psi \varphi = -\int [\exp(-h)(-\phi \nabla h  +
\nabla \phi )]\varphi= - \int \nabla \phi  \exp(-h) \varphi+ \int
\exp(-h)\varphi \phi \nabla h \\=  \int \phi  [ \exp(-h) \nabla
\varphi- \nabla h \exp(-h)\varphi]+ \int \exp(-h)\varphi \phi \nabla
h = \int \psi \nabla \varphi.
\end{eqnarray*}
Inequality (\ref{secondstep}) follows from (\ref{weakbound}).
\end{proof}

\begin{proof}[Proof of Proposition \ref{cllevel}]
Since $V$ is continuous and diverges at the boundaries, it must be
bounded from below. Also, it follows from Assumption
\ref{adiffusionb} that
\begin{eqnarray*}
{\mathcal V}_\theta \subset \{ \psi \in L^2(leb) : & \psi &  {\rm
has \ \ a \ \ weak
\ \ derivative \ \ and} \\
& \ \ &  \int \left(\theta   + {\frac 1 2} V \right) (\psi)^2 +
{\frac { \underline c} 2 }  \int |\nabla \psi|^2 \le 1 \}.
\end{eqnarray*}
We may then apply the argument in the proof of Theorem XIII.67 of
\citet{reedsimon} to establish that ${\mathcal V}_\theta$ is
precompact in ${L}^2(leb)$.
\end{proof}


\begin{proof}[Proof of Lemma \ref{lemma:davies}]
Consider a positive function
\[
\chi(x) = {\frac 1 \varsigma},
\]
and note that
\[
\left[\varsigma(x)^2 - {\underbar {\sf c}}\right]\nabla \chi(x) = -
\varsigma(x) {\nabla v}  + {\underbar {\sf c}} {\frac {\nabla
v(x)}{\varsigma(x)}}.
\]
For $\phi$ in $C_K^2$ we may apply integration by parts to show that
\begin{eqnarray*}
\int (\varsigma^2 - {\underbar {\sf c}}) \nabla \chi \cdot \nabla
\phi & = & \int \left[\left( {\frac {\varsigma^2 + {\underbar {\sf
c}}} \varsigma } \right) (\nabla v \cdot \nabla v) + \left({\frac
{\varsigma^2 - {\underbar {\sf c}}} \varsigma}\right) {\rm
trace}\left({\frac {\partial^2 v} {\partial x_i
\partial x_j}}\right) \right] \phi \cr
& = & \int {\check W}\chi \phi
\end{eqnarray*}
The conclusion follows from Theorem 1.5.12 of \citet{davies89}. While
\citet{davies89} uses test functions $\phi$ in $C_K^\infty$, the same
proof applies to $C_K^2$ test functions.
\end{proof}

\begin{proof}[Proof of Proposition \ref{rev-markov}]
The form $f$ satisfies the Beurling-Deny criteria (\citet{davies89}
Theorem 1.3.5). Thus there exists a self-adjoint operator $F$ which
is an extension of $F_o$ and generates a submarkov semigroup
$\exp(-tF)$. Theorem 7.2.1 of \citet{fot} guarantees that there
exists a Markov process $\{x_t\}$ that has $\exp(-tF)$ as its
semigroup of conditional expectations.  The semigroup $\exp(-tF)$
conserves probability because the unit function is in the domain of
the form $f$ and $f(1,\phi)=0$ for any $\phi \in {\mathcal D}(f)$.
As a consequence, the unit function is also in the domain of the
operator $F$, $F 1 =0$.
\end{proof}

\end{appendix}

\bibliography{longrun2}

\end{document}